\documentclass[11pt, letterpaper]{article}
\usepackage{graphicx} % Required for inserting images

\usepackage[english]{babel}
\usepackage[utf8]{inputenc}
\usepackage[T1]{fontenc}
\usepackage[a4paper, margin=1in]{geometry}
\usepackage{graphicx}
\usepackage{framed}
\usepackage[framemethod=tikz]{mdframed}
\usepackage{todonotes}
\usepackage{color}

\usepackage{xcolor}

\definecolor{darkgreen}{rgb}{0,0.5,0}
\definecolor{darkgray}{rgb}{0.2,0.2,0.2}
\usepackage{hyperref}
\hypersetup{
    unicode=false,          % non-Latin characters in Acrobat’s bookmarks
    colorlinks=true,        % false: boxed links; true: colored links
    linkcolor=blue,         % color of internal links (change box color with linkbordercolor)
    citecolor=purple,       % color of links to bibliography
    filecolor=magenta,      % color of file links
    urlcolor=cyan           % color of external links
}

\usepackage{amsthm}
\usepackage{amsmath}
\usepackage{amssymb}
\usepackage{amsfonts}
\usepackage{mathrsfs}
\usepackage{mathtools}
\usepackage{verbatim}
\usepackage{footnote}

\usepackage[linesnumbered,boxed, vlined]{algorithm2e}
	\DontPrintSemicolon
  \SetKwInOut{Input}{Input}
  \SetKwInOut{Output}{Output}
  \SetKwProg{function}{Function}{:}{}

\usepackage{algorithmic}

\usepackage{lineno}
\usepackage{caption}
\usepackage{framed}
\usepackage{enumerate}

\usepackage{tikzsymbols}
\usepackage{thmtools,thm-restate}
\usepackage{nicefrac}

\usepackage{subcaption}
\usepackage{pdfpages}
\usepackage{makecell}

\usepackage[capitalize, nameinlink]{cleveref}
\usepackage[section]{placeins}

\newtheorem{theorem}{Theorem}[section]
\newtheorem{lemma}[theorem]{Lemma}
\newtheorem{corollary}[theorem]{Corollary}

\newtheorem{observation}[theorem]{Observation}

\newtheorem{proposition}{Proposition}
\newtheorem*{remark*}{Remark}

\newcommand{\R}{\mathbb{R}}

\newcommand{\bX}{\mathbf{X}}
\newcommand{\bY}{\mathbf{Y}}

\newcommand{\hll}{\hat{\ell}}

\newcommand{\ex}[1]{\mathop{{\mathbb E}\left[ #1 \right]}}

\def\polylog{\operatorname{polylog}}

\def\poly{\operatorname{poly}}

\newcommand{\lone}{\ensuremath{L_1}\xspace}

\newcommand{\out}[1]{}

\newcommand{\prob}[1]{\Pr \left[ #1 \right]}
\newcommand{\eps}{\epsilon}

\newcommand{\tpi}{\tilde{\pi}}

\newcommand{\piadd}{\pi_{\textsc{add}}}

\title{Dynamic PageRank: Algorithms and Lower Bounds} %TODO Please add

\author{{Rajesh Jayaram}\\ {Google Research, New York, USA}\\ \href{mailto:rkjayaram@google.com}{rkjayaram@google.com} 
\and
{Jakub Łącki}\\ {Google Research, New York, USA} \\
\href{mailto:jlacki@google.com}{jlacki@google.com}
\and
{Slobodan Mitrovi\' c\thanks{Supported by the Google Research Scholar and NSF Faculty Early Career Development Program \#2340048.}}\\{UC Davis, Davis, USA}\\
\href{mailto:smitrovic@ucdavis.edu}{smitrovic@ucdavis.edu}
\and
{Krzysztof Onak}\\{Boston University, Boston, USA} \\
\href{mailto:krzysztof@onak.pl}{krzysztof@onak.pl}
\and
{Piotr Sankowski\thanks{Partially supported by the National Science Center (NCN) grant no. 2020/37/B/ST6/04179 and the ERC CoG grant TUgbOAT no 772346.}} \\ IDEAS NCBR, University of Warsaw and \\ MIM Solutions, Warsaw, Poland \\ \href{mailto:sank@mimuw.edu.pl}{sank@mimuw.edu.pl}
}

\begin{document}

\maketitle

\begin{abstract}
We consider the PageRank problem in the dynamic setting, where the goal is to explicitly maintain an approximate PageRank vector $\pi \in \R^n$ for a graph under a sequence of edge insertions and deletions.  
Our main result is a complete characterization of the complexity of dynamic PageRank maintenance for both  multiplicative and additive ($L_1$) approximations. 

First, we establish matching lower and upper bounds for maintaining additive approximate PageRank in both incremental and decremental settings. In particular, we demonstrate that in the worst-case $(1/\alpha)^{\Theta(\log \log n)}$ update time is necessary and sufficient for this problem, where $\alpha$ is the desired additive approximation. On the other hand, we demonstrate that the commonly employed ForwardPush approach performs substantially worse than this optimal runtime. Specifically, we show that ForwardPush requires $\Omega(n^{1-\delta})$ time per update on average, for any $\delta > 0$, even in the incremental setting.

For multiplicative approximations, however, we demonstrate that the situation is significantly more challenging.
Specifically, we prove that any algorithm that explicitly maintains a constant factor multiplicative approximation of the PageRank vector of a directed graph must have amortized update time $\Omega(n^{1-\delta})$, for any $\delta > 0$, even in the incremental setting, thereby resolving a 13-year old open question of Bahmani et al.~(VLDB 2010). This sharply contrasts with the undirected setting, where we show that $\poly \log n$ update time is feasible, even in the fully dynamic setting under oblivious adversary.
\end{abstract}

\section{Introduction}
The notion of PageRank was introduced by Brin and Page 25 years ago to rank web search results~\cite{brin1998anatomy}.
Since then, computing the PageRank of a network has become a fundamental task in data mining~\cite{Wu2008}.
At a high level, PageRank is a probability distribution over the vertices of a directed graph which assigns higher probability to more ``central'' vertices; see \cref{sec:preliminaries} for a formal definition. 
We write $\pi \in \R^n$ to denote PageRank probability vector, where $\pi_i$ is the probability mass on the $i$-th vertex.
Due to its importance, it has been studied extensively in a number of computational models.
In this paper, we consider the PageRank problem in the dynamic setting, in which the goal is to maintain an approximate PageRank vector $\tilde{\pi} \in \R^n$ of a graph undergoing a sequence of edge insertions and deletions. We focus primarily on explicit maintainance of the PageRanks, meaning that the algorithm explicitly maintains $\tilde{\pi}$ in its memory contents at all time steps; we remark that all prior algoriths for the problem of maintaining all PageRanks in the dynamic setting have been of this form.

%For example, fully incremental scenario, i.e., when edges are only added to the graph, is very natural as many networks grow over time.

We consider \emph{three} different settings, which differ in the allowed sets of operations.
In the \emph{incremental} setting, edges can only be added to the graph.
Analogously, in the \emph{decremental} setting, edges can only be deleted.
The most general setting is the \emph{fully dynamic} setting in which we allow both types of updates.
We also consider two notions of approximation.
A \emph{$1+\alpha$ multiplicative approximation} to the PageRank vector $\pi$ is a vector $\tilde{\pi}$, such that for every vertex $v$ it holds $\tpi_v \in [(1-\alpha)\pi_v, (1+\alpha)\pi_v]$.
An \emph{additive $\alpha$ approximation} is a vector $\tilde{\pi}$ such that $\|\tilde{\pi} - \pi\|_1 \leq \alpha$.\footnote{Note that this coincides with the total variational distance between distributions.}
We note that a multiplicative guarantee is strictly stronger, as a multiplicative $1+\alpha$ approximation implies an additive $\alpha$ approximation.

Previous work on dynamic PageRank~\cite{bahmani2010fast, liao2017monte, zhan2019fast, bahmani2012pagerank, 10.1145/2701126.2701165} resulted in two main approaches to the problem.
The first one is based on sampling random walks.
Specifically, it is well-known that one can approximate PageRank by sampling $O(\log n)$ random walks of length $O(\log n)$ from each vertex in the graph (see \cref{alg:PageRank-approximation}).

In a seminal paper, Bahmani et al.~\cite{bahmani2010fast} showed that this approach can be made dynamic.
Specifically, the algorithm of Bahmani et al.~maintains a \emph{multiplicative} $1+\alpha$ approximation of incremental (or decremental) PageRank when the updates \emph{arrive in a random order}. 
However, their analysis crucially relies on the random arrival of updates, and it was not clear whether this assumption could be removed. 
The authors of~\cite{bahmani2010fast} explicitly posed the question of whether it is possible to extend their results for multiplicative approximations to the case of adversarially ordered updates; to date, this question has remained open.

%The algorithm handles any sequence of updates in time which is near-linear in the size of the graph, as long as the operations are given in random order.
%The algorithm maintains an approximate PageRank vector \emph{explicitly}, meaning that after each update to the graph it modifies its internal representation of $\tpi$ to be a $1+\alpha$ approximation.

%which means that it maintains a probability distribution $\tpi$, which is a $1+\alpha$ approximation of the exact PageRank, and after each update outputs the list of changes to $\tpi$.

The second approach to computing dynamic PageRank is a dynamic version of the ForwardPush algorithm~\cite{zhang2016approximate,aggarwal2021performance,guo2017parallel}, which is a variant of a classical local push approach proposed by \cite{andersen2006local}.
This algorithm was developed for the problem of maintaining Personalized PageRank, but can also be naturally used to maintain an additive PageRank approximation.
While this approach is highly effective in practice, no running time bounds faster than running a static algorithm from scratch after each update have been developed for maintaining PageRank using the dynamic ForwardPush method.\footnote{We note that the paper introducing the dynamic ForwardPush algorithm gives a good running time bound for running the algorithm in \emph{undirected} graphs. However, this bound only holds for computing Personalized PageRank from a \emph{uniformly random} source vertex. Even though PageRank can be reduced to Personalized PageRank, the reduction requires computing Personalized PageRank from a \emph{fixed} vertex, and so the bound does not carry over.}

Thus, despite the above line of work, many fundamental questions regarding the computational cost of maintaining PageRank in a dynamic setting remain open.
Specifically, it is still open whether there exists an algorithm for maintaining a approximation to PageRank in $o(n)$ time per update. 
This question is open even if one considers only incremental or decremental updates, or  if one allows additive approximation.
In this paper, we answer each of these open questions. More precisely, we characterize the complexity of solving the dynamic PageRank problem in each of these settings by providing new upper and lower bounds.

%Existing algorithms for maintaining dynamic PageRank either make simplifying assumptions on the allowed operations~\cite{bahmani2010fast}, or do not provide strong theoretical guarantees for either the running time or the approximation quality~\cite{liao2017monte,zhan2019fast}.
%In particular, to the best of our knowledge, no per-update runtime bounds sublinear in the number of vertices
%have been achieved for algorithms that maintain a provable PageRank approximation in the general setting, even if the allowed update operations are limited to either only edge insertions or only edge deletions.

%In this paper we close this gap in the literature by . Specifically, we show that there exists an efficient algorithm for maintaining \emph{additive} approximation to the PageRank.
%On the negative side, we prove that explicitly maintaining a constant \emph{multiplicative} approximation in time which is sublinear in the number of vertices is hard.

 %Further, we demonstrate that if \textit{additive} error to the PageRanks is allowed (i.e., total variational distance error in the distribution), then efficient algorithms do exist. 

\subsection{Our contributions}

We provide new lower and upper bounds on the complexity of explicitly maintaining an approximate PageRank vector both under additive and multiplicative approximation.
Throughout this section, we use $n$ to denote the number of vertices in a graph, $m$ to denote the number of edges and $\epsilon$ to denote the jumping probability used to define PageRank.\footnote{The probability of not-jumping (in our notation, $1-\epsilon$) is sometimes called the \emph{damping factor} of PageRank.}

\subsubsection{Additive Approximation}
We provide (essentially) matching lower and upper bounds for explicitly maintaining additive approximation of PageRank in both incremental and decremental setting.

\begin{restatable}{theorem}{lbadditive}\label{thm:LB-additive}
Fix $\epsilon \in (0.01, 0.99)$. For any sufficiently large $n \geq 1$ and any $\alpha$ such that $1/\alpha = n^{o(1 / \log \log n)}$, any algorithm which explicitly maintains $\alpha$-additive approximation of PageRank must run in $n \cdot (1/\alpha)^{\Omega(\log \log n)}$ total time.
\end{restatable}

Our lower bound, which we prove in \cref{sec:LB-additive-approx}, is obtained by constructing a graph and an update sequence for which the PageRank vector undergoes a large number of significant changes.
The changes to the vector are large to the point that even an approximate PageRank vector must be often updated in linear time.
We note that the lower bound, and all other lower bounds that we state, applies to the setting when the PageRank vector is maintained explicitly, i.e., after each update algorithm outputs the changes that the PageRank vector undergoes.

We note that it is easy to come up with an example in which a single edge update significantly changes the PageRanks of a large fraction of vertices (see \cref{fig:worst-case}).
This immediately rules out efficient incremental and decremental algorithms that maintain approximate PageRank with \emph{worst-case} update time guarantees.
This also rules out fully dynamic algorithms with amortized update time guarantees.
However, proving a strong lower bound for the \emph{amortized} update time bound in the incremental or decremental setting is far more involved, as it requires showing a long sequence of updates in which, on average, every edge insertion (or deletion) changes the PageRank of many vertices.

We complement our lower bound with the following algorithmic result proved in \cref{sec:improved-bounds}.

\begin{restatable}{theorem}{directedadditiveup}
\label{thm:bounded-length}
For any $\epsilon \in (0, 1)$, there is an algorithm that with high probability explicitly maintains an $\alpha$ additive approximation of PageRank of any graph $G$ in either incremental or decremental setting.
The algorithm processes the entire sequence of updates in $O(m) + n \cdot (1/\alpha)^{O_\eps(\log \log n)}$ total  time and works correctly against an oblivious adversary.
\end{restatable}

Furthermore, we study the complexity of the dynamic ForwardPush algorithm~\cite{zhang2016approximate}.
This algorithm, when run with parameter $\tilde{\alpha}$ maintains an $\tilde{\alpha} \cdot m$ additive approximation to PageRank (and so to obtain $\alpha$ additive approximation, one needs to use $\tilde{\alpha} = \alpha / m$).
By using a similar construction of a hard instance, we show that the algorithm takes $\Omega(n^{2-\delta})$ time, for any $\delta > 0$, to handle a sequence of $O(n)$ operations, even in incremental or decremental settings (see \cref{thm:fp}).

\subsubsection{Multiplicative Approximation}
Our next result is a lower bound showing that any algorithm explicitly maintaining a constant multiplicative approximation to PageRank, even in the incremental or decremental setting, must in the worst case take $\Omega(n^{2-\delta})$ total time, for any $\delta > 0$, to process a sequence of $n$ updates to an $n$-vertex graph. Specifically, we prove the following in \cref{sec:LB-multiplicative-approx}:

\begin{restatable}{theorem}{thmmub}\label{thm:lb-explicit}
There exists a sequence of $\Theta(n)$ edge insertions applied to an initially empty graph on $n$ vertices for which the following holds. For any constant $\delta > 0$, any algorithm that maintains a vector $\tilde{\pi} \in \R^n$ such that $(1/2) \pi_v < \tilde{\pi}_v \leq 2 \pi_v$ at all time steps, must take time $\Omega(n^{2-\delta})$ to process the sequence.
In particular, the amortized update time of any such algorithm is $\Omega(n^{1-\delta})$. 
\end{restatable}

We note that, by symmetry, the above theorem also applies to the decremental setting.

\cref{thm:lb-explicit} gives a negative resolution to the 13-year-old open question of Bahmani et al.~\cite{bahmani2010fast}, who asked whether their polylogarithmic update time bounds for maintaining PageRank under a sequence of updates coming in \emph{random} order can be extended to the general case.
Previously, the only negative results for this problem were given by Lofgren~\cite{DBLP:journals/ipl/Lofgren14} who showed that the specific algorithm of Bahmani et al.~requires $\Omega(n^c)$ update time for some $c \in (0, 1)$, but this did not rule out the existence of a better algorithm.
We extend this lower bound to \emph{every} algorithm which explicitly maintains an approximate PageRank vector, and strengthen the bound from $\Omega(n^c)$ to $\Omega(n^{1-\delta}$) for any $\delta>0$.

%Our lower bound construction consists of an $n$ vertex graph with $O(n)$ edges and an order of deleting edges from the graph, which result in overall $\Omega(n^{2-o(1)})$ significant changes (i.e., changes by more than a constant factor) to the PageRank vector.
%This immediately shows that explicitly maintaining constant multiplicative PageRank approximation requires $\Omega(n^{1-o(1)})$ amortized time.

\begin{comment}
\begin{figure}[t]
    \centering
    \includegraphics[scale=0.5]{figs/PageRank-example.eps}
    \caption{An example illustrating that a single edge insertion/deletion may significantly change PageRank values of multiple vertices.
    In particular, inserting edge $uv$ changes the PageRanks of $n / \log n$ vertices in $R$ from $O(1/n)$ to $O(\log n / n)$.}
    \label{fig:worst-case}
\end{figure}
\end{comment}

\iffalse

In sharp contrast with our first lower bound, we demonstrate how to bypass the hardness of Theorem~\ref{thm:lb-explicit} by using additive approximation instead of multiplicative.
Specifically, we give a new analysis of (a slight extension of) the algorithm of Bahmani et al.~\cite{bahmani2010fast}, and show that it maintains a constant additive approximation of PageRank in polylogarithmic amortized update time in both  the incremental and decremental settings (\cref{thm:bounded-length}), even if the updates are given in arbitrary order.
The algorithm works correctly against an oblivious adversary.
\fi

To complement the above lower bound, in \cref{sec:upper-bound-undirected}, we give a simple analysis of the Bahmani et al.~algorithm in \emph{undirected} graphs, and show that in this case maintaining multiplicative approximation can be done in polylogarithmic time per update even in the fully dynamic setting.
This algorithm also assumes an oblivious adversary.
While the analysis is based on a simple observation, to the best of our knowledge it has not been explicitly given before.

\begin{restatable}{theorem}{thmundirected}
\label{thm:undirected}
For any $\epsilon \in (0, 1)$, there is an algorithm that with high probability explicitly maintains a $1+\alpha$ multiplicative approximation of PageRank of any undirected graph $G$ in the fully dynamic setting.
The algorithm handles each update in $O(\log^5 n / (\epsilon^2 \alpha^2))$ time and works correctly against an oblivious adversary.
\end{restatable}

It is an open question whether it is possible to design dynamic PageRank algorithms that bypass our lower bounds, for example, by not maintaining PageRank explicitly or looking beyond worst-case bounds and studying restricted graph classes.

\subsection{Related Work}
The dynamic PageRank problem has been studied in a number of recent works~
\cite{bahmani2010fast,bahmani2012pagerank,liao2017monte,zhan2019fast, chakrabarti2007dynamic, pathak2008index, rossi2012dynamic, lofgren2016personalized, sahu2022dynamic, sallinen2023real, guo2017parallel} studying both the theoretical and empirical aspects of the problem.
One line of study considered the incremental and decremental settings with updates performed in \emph{random} order~\cite{bahmani2010fast,zhan2019fast} and obtained algorithms that achieve $O(\log n / \eps)$ update time. The result of \cite{zhan2019fast} is applicable in a non-random order as well, although in that case it requires $\Omega(d_v)$ running time per update done on a vertex $v$ of degree $d_v$. Bahmani et al.~\cite{bahmani2012pagerank} analyze their algorithm in a random graph model in which high PageRank vertices are more likely to receive new neighbors.
We note that attempts at designing faster algorithms have been undertaken in \cite{liao2017monte} as well as \cite{zhan2019fast}. However, these algorithms come with no provable approximation guarantees. 
%Their algorithms should be considered heuristics, since they have been given without any correctness proof. Additionally, these approaches repeat the problem present in~\cite{bahmani2010fast} and regenerate walks from scratch without remembering their lengths. Moreover, they introduce modifications that lead to further accumulation of errors. For example, in~\cite{liao2017monte} the walks are not being stored and only the counters of the numbers of walks are updated, whereas in~\cite{zhan2019fast} walks are being regenerated starting from different vertices. To the best of our understanding, these design decisions do not come with provable approximation guarantees, as opposed to the algorithms we study in our paper. 

Another line of work~\cite{andersen2007local,lofgren2013personalized,lofgren2014fast,zhang2016approximate,wang2020personalized} focuses on computing Personalized PageRank, which is PageRank computed from the point of view of a single vertex. For instance, \cite{lofgren2014fast} show that if each entry of a Personalized PageRank is lower-bounded by $\delta$, then the Personalized PageRank of a vertex can be approximated in time $O(\sqrt{d / \delta})$, where $d$ is the average graph degree.

Finally, PageRank was also studied in the context of sublinear algorithms~\cite{bressan2018sublinear, stoc24}.
For instance, for a graph on $m$ edges and omitting $\poly$ dependence on $\log m$ and $\alpha^{-1}$, the very recent algorithm presented in \cite{stoc24} requires $O(n^{1/2} \cdot \min\{m^{1/4}, \Delta^{1/2}\})$ running time for approximating the PageRank of a single vertex, where $\Delta$ is the maximum degree in the graph.

%\jltodo{We should cite \url{https://ieeexplore.ieee.org/document/8555151} which is probably a good method if we want to maintank PageRank of a small number of vertices}
% \stodo{Done.}

\subparagraph{Further related work.}
Maintaining random walks and studying their properties in the context of dynamic graphs has been done for a number of problems. 
In addition to PageRank, the examples include distributed algorithms~\cite{augustine2016distributed}, vertex sparsifiers~\cite{durfee2019fully,van2022faster}, max-flow~\cite{van2022faster}, and cover and mixing time~\cite{avin2018cover,sauerwald2019random}.

\begin{comment}

\noindent
\textbf{Personalized PageRank and other computational models. }

Designing distributed and parallel algorithms for PageRank computation, although not necessarily in the dynamic setting, was also a focus of a number of papers~\cite{guo2017parallel,bahmani2011fast,das2013fast,lin2019distributed,lkacki2020walking}. The works in this direction typically focus on minimizing the communication needed to approximate PageRank.

%\stodo{Will also mention local computation algorithms. Anything else?}\kotodo{Maybe not needed or just add to the previous paragraph? What is the citation?}

\end{comment}

\subsection{Impossibility of Non-Trivial Worst-Case Bounds}\label{sec:impossibility}
\begin{figure}
    \centering
    \includegraphics[scale=0.7]{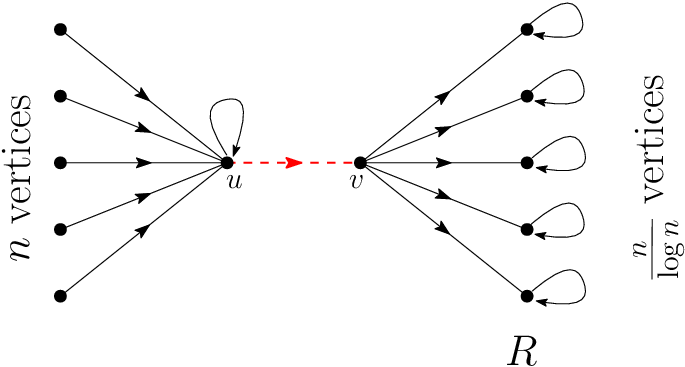}
    \caption{An example illustrating that maintaining multiplicative approximation or even an $\lone$ approximation of PageRank in the worst case requires $\Theta(n / \log n)$ running time even after a single deletion/insertion of edge $uv$.
    For details, see \cref{sec:impossibility}.}
    \label{fig:worst-case}
\end{figure}
A wealth of literature on designing dynamic algorithms for approximate PageRank, including our results, focuses on \emph{amortized} running time complexity.
It is natural to wonder whether non-trivial worst-case update running times do not exists due to lack of techniques or due to fundamental reasons. 
As our example in \cref{fig:worst-case} illustrates, non-trivial update running times are not possible even on very sparse graphs and even if one's goal is to maintain an $L_1$-approximate PageRank vector.

Namely, on the one hand, for the graph $G$ in \cref{fig:worst-case}, it can be shown that $\pi_u, \pi_v \in \Omega(\eps)$ and $\pi_x \in \Omega((\log n) / n)$ for each vertex $x \in R$. 
On the other hand, consider graph $G'$ obtained from $G$, i.e., from the graph in \cref{fig:worst-case}, by removing the red-dashed $(u, v)$ edge, and let $\pi'$ be the PageRank of $G'$. It is not hard to show that $\pi'_u \in \Omega(1)$, $\pi'_v \in O(\eps / n)$ and $\pi'_x \in O(1/n)$ for each $x \in R$. This example illustrates the following: there exists a directed graph in which after \emph{a single} edge removal one has to update $\Omega(n / \log n)$ vertices if the goal is to maintain a multiplicative and even an $L_1$ approximation of the PageRank for sufficiently small constant $\eps$.
Moreover, if random walks are used to estimate the PageRank -- which to the best of our knowledge is the only other used approach than Power method -- then maintaining an additive or multiplicative approximate PageRank of a single vertex still requires $\Omega(n)$ worst-case time. To see that, observe that there are $\Theta(n)$ times more random walks passing through $v$ in $G$ than in $G'$.

\subsection{Organization of the Paper}
The rest of this paper is organized as follows.
In \cref{sec:preliminaries} we formally define PageRank and review a random-walk based algorithm for approximating it in the static setting.
In \cref{sec:lower-bound-insertion-only} we give the lower bounds on the time required to explicitly maintain PageRank and on the running time of the dynamic ForwardPush algorithm.
\cref{sec:random-walks} reviews the algorithm for approximating PageRank by maintaining random walks.
While the algorithm is essentially the same as the algorithm by Bahmani et al.~\cite{bahmani2010fast}, we present a full analysis, since the previous papers on dynamic PageRank did not prove the correctness of this approach.
In the following two sections we analyze this algorithm in two settings.
First, in \cref{sec:improved-bounds} we show that this algorithm achieves near-optimal update time while maintaining additive approximation to PageRank.
Second, in \cref{sec:upper-bound-undirected} we present a simple analysis showing that in undirected graphs maintaining even a constant multiplicative approximation to PageRank in the fully dynamic setting is possible with polylogarithmic update time.
\section{Preliminaries}
\label{sec:preliminaries}

We begin by defining the PageRank of a directed graph $G = (V,E)$. Formally, the PageRank of $G$, denoted by $\pi \in \R^n_{\geq 0}$, is the stationary distribution of a random walk on $G$, where at each step the walk \textit{jumps} to another uniformly random vertex with probability $\eps \in (0,1)$. The jump probability $\eps$ is a parameter, which we will fix for the remainder. If $\deg(i)$ is the out-degree of the $i$-th vertex in $V$, then the corresponding non-symmetric transition matrix $M \in \R^{n \times n}$ has entries $M_{i,j} = \frac{\eps}{n} + (1-\eps)\frac{1}{\deg(i)}$ if $(i,j) \in E$, and $M_{i,j} = \eps/n$ otherwise. We make the standard assumption (required for PageRank to be well-defined) that each vertex has $\deg(i) \geq 1$, which can be accomplished by adding self-loops. 

PageRank can be approximated by sampling $O(\log n / \eps)$ relatively short random walks from each vertex. One such approach is provided as \cref{alg:PageRank-approximation}, for which the following can be shown.
\begin{proposition}[\cite{bahmani2010fast,lkacki2020walking}]\label{prop:walksEstimate}

Let $\pi$ be the PageRank vector of a graph $G$.
The estimate $\tilde{\pi}$ computed by \cref{alg:PageRank-approximation} (with $\ell = \infty$) satisfies (a) for all $v \in V$ we have $\ex{\tilde\pi_v} = \pi_v$, and (b) with probability $1-1/\poly(n)$, simultaneously for all $v \in V$, we have $\tilde{\pi}_v= (1 \pm \alpha)\pi_v$.
\end{proposition}

\begin{algorithm}[t]
\begin{algorithmic}[1]
	\STATE{Sample a set $W$ of $R = \left \lceil \frac{9\ln n }{\eps \alpha^2} \right \rceil$ random walks starting from each vertex of $G$. \\
	Each walk length is chosen from geometric distribution with parameter $1-\epsilon$.}
    \STATE{Remove from $W$ all walks longer than $\ell$.}
\FOR{$v\in V$}
%\State{Let $n_v$ be the number of the walks from $W$ ending in $v$.}
    \STATE{$\bX_v \gets $ the number of times the walks from $W$ visit $v$.}
    %\State{$\tilde{\pi}(v) \gets \frac{n_v}{Kn}$.}
    \STATE{$\tilde{\pi}(v) \gets \frac{\bX_v}{|W| / \epsilon}$.\label{line:tpi-assign}}
\ENDFOR
\STATE{Return $\tilde{\pi}$}
\end{algorithmic}
\caption{An algorithm for computing approximate PageRank using random walks. \textbf{Input:} A graph $G$, and parameters $\eps$, $\alpha$ and $\ell$.}
\label{alg:PageRank-approximation}
\end{algorithm}

\section{Lower Bounds}
\label{sec:lower-bound-insertion-only}

In this section we present our lower bounds for maintaining explicit approximation to PageRank and for the running time of the dynamic ForwardPush algorithm~\cite{zhang2016approximate}.
We now describe a generic construction of a hard instance, which we instantiate with different parameters in each of the individual lower bounds.
Throughout the section, we consider the case of $\epsilon \in (0.01, 0.99)$, which is the usual case in the applications of PageRank.

\paragraph{The Graph.}
The graph $G$ is a union of the graphs $H,R,S_0,S_1$. First, $H$ is a directed tree.
Each non-leaf vertex $v$ has exactly $t$ children, with $p^i$ parallel directed edges from $v$ to the $i$-th child of $v$ (where $i$ is $0$-based).
We require that $p \geq \max(1/\epsilon, 2)$. Hence, the total-out degree of each internal vertex in $H$ is $O(p^t)$.
The depth of $H$ is $d$, and so $H$ has $\Theta(t^d)$ vertices and $\Theta(t^d \cdot p^t)$ edges.

The graph $R$ consists of $n/4$ vertices $v$, each with no in-edges, and each with a single out edge $vr$ where $r$ is the root of the directed tree $H$. Finally, the sets $S_0,S_1$ are both directed star graphs on $s+1$ vertices (with the edges directed away from the center of the star), where $S_i$ has the center $c_i$ for $i \in \{0,1 \}$. 
Additionally, each leaf of $S_0$ and $S_1$ has a single outgoing edge, which is a self-loop.
We then order the leaf vertices of $H$ as $\ell_1,\ell_2,\dots$, and create a directed edge from $\ell_i$ to $c_{i \bmod 2}$.
We will set the parameters, such that the total number of vertices in $H$, $R$, $S_0$, and $S_1$ is less than $n$.
One can then add an additional $O(n)$ isolated vertices (with self-loops), so that the total number of vertices is precisely $n$.

\paragraph{Update Sequence} The initial graph has all vertices and edges of $H, R, S_0, S_1$, except that each non-leaf vertex of $H$ only has an edge to its leftmost child (i.e., one with index $0$).
Observe that each vertex has at least one outgoing edge, and so PageRank is well-defined.

The update sequence is as follows. Let $v_1,\dots,v_{|H|}$ be the sequence of vertices visited on a pre-order traversal of $H$, such that $\ell_1,\ell_2,\dots$ is a subsequence of $v_1,\dots,v_{|H|}$.
We insert the edges of $H$ in $|H|$ rounds: in the $i$th round we insert all incoming edges of $v_i$ (unless they have already been in the graph from  the beginning).

To prove the lower bounds, we use the following way of interpreting PageRank, which is a continuous version of~\cref{alg:PageRank-approximation} and follows from \cref{prop:walksEstimate}.
Each vertex has some \emph{probability mass}, which it either generates or receives from its in-neighbors.
Specifically, each vertex of the graph generates a probability mass of $1/n$.
A $1-\epsilon$ fraction of the probability mass of a vertex $v$ (either generated by $v$ or incoming to $v$ from other vertices) is divided uniformly among the outgoing edges of $v$ and sent to the neighbors of $v$.
The PageRank of each vertex is exactly $\epsilon$ fraction of its probability mass.

Note that if a vertex is on a cycle, some probability mass enters it multiple times.
In this case, each time the mass enters the vertex, it increases the total probability mass.
In particular, we have the following.

\begin{observation}\label{obs:loop}
Let $v$ be a vertex, whose only outgoing edge is a self loop.
Assume that $v$ receives a probability mass of $p$ along its incoming edges other than the self-loop.
Then, the PageRank of $v$ is $p + 1/n$.
\end{observation}

\begin{lemma}\label{lem:pagerank-leaf}
Consider the graph $G^\tau$ obtained right after inserting all edges on the path from $R$ to $\ell_i$.
Let $m_i$ be the probability mass that reaches $\ell_i$ from $R$ in $G^\tau$.
Then $m_i \geq (1-\epsilon)^{2d+2}/4$.

Moreover, out of the probability mass that reaches the leaves of $H$ from $R$, at least $(1-1/p)^d$ fraction reaches $\ell_i$.
\end{lemma}

\begin{proof}
Observe that a path from any vertex $u \in R$ to $\ell_i$ first follows the edge to $r$, which is the only outgoing edge of $u$, and then, thanks to the order of adding edges of $H$, at each step uses the rightmost edge of each vertex in $H$.
Consider an internal vertex $w \in H$.
By the construction it has $p^i$ edges to the $i$th child (0-based).
Assuming that we have added edges to $j$ children so far, we have that there are $p^{j-1}$ edges to the rightmost child and so the fraction of outgoing edges of $w$ that go to the rightmost child is:
\begin{equation}\label{eq:rightchild}
p^{j-1} / \left (\sum_{k=0}^{j-1} p^k\right) = p^{j-1} \cdot \frac{p-1}{p^j - 1} \geq p^{j-1} \cdot \frac{p-1}{p^j} = 1-1/p.
\end{equation}
The path from $w$ to $\ell_i$ has $d+1$ edges.
At each step $1-\epsilon$ fraction of the probability mass is forwarded to the children, out of which, as shown above, at least $1 - 1/p \geq 1-\epsilon$ fraction follows the path to $\ell_i$.
Hence, the fraction of probability mass that reaches $\ell_i$ from $w$ is $(1-\epsilon)^{2d+2}$.
Since vertices of $R$ generate a total probability mass of $1/4$, we get the desired.

The second claim follows directly from \cref{eq:rightchild} and the fact that $H$ has depth $d$.
\end{proof}

\subsection{Lower Bound for Maintaining Additive Approximation}
\label{sec:LB-additive-approx}

We first show the following auxiliary lemma which we will use to argue when an additive $\alpha$-approximate PageRank vectors must be updated in linear time.

\begin{lemma}\label{lem:coordinates}
Consider four vectors $v^1, \tilde{v}^1, v^2, \tilde{v}^2 \in \mathbb{R}^n$, such that $\|v^1 - \tilde{v}^1\|_1 \leq \alpha$,
$\|v^2- \tilde{v}^2\|_1 \leq \alpha$ and $v^1$ and $v^2$ differ by at least $100 \cdot \alpha / n$ on at least $n/4$ coordinates.
Then $\tilde{v}^1$ and $\tilde{v}^2$ differ on $\Omega(n)$ coordinates.
\end{lemma}

\begin{proof}
The proof goes by contradiction.
Assume that $\tilde{v}^1$ and $\tilde{v}^2$ differ on at most $n/1000$ coordinates.
Thus, they have at least $0.999 \cdot n$ coordinates in common.
Moreover, $\|v^1 - \tilde{v}^1\|_1 \leq \alpha$ implies that $v^1$, and $\tilde{v}^1$ differ by more than $10 \cdot \alpha / n$ on less than $0.1 \cdot n$ coordinates.
Clearly, a similar property is satisfied by $v^2$, and $\tilde{v}^2$.

Let $I$ be the set of coordinates where
\begin{enumerate}
    \item $\tilde{v}^1$ and $\tilde{v}^2$ are equal (there are at least $0.999 \cdot n$ such coordinates),
    \item $v^1$ and $v^2$ differ by at least $100 \cdot \alpha / n$ (at least $n/4$ such coordinates),
    \item $v^1$ and $\tilde{v}^1$ differ by at most $10 \cdot \alpha / n$ (at least $0.9 \cdot n$ coordinates),
    \item $v^2$ and $\tilde{v}^2$ differ by at most $10 \cdot \alpha / n$ (at least $0.9 \cdot n$ coordinates).
\end{enumerate}

Observe that since the vectors have $n$ coordinates, $I$ is nonempty.
By using first the triangle inequality, and then items 2-4 above, for any coordinate $i \in I$ we have 
\begin{align*}
|\tilde{v}^1_i - \tilde{v}^2_i| & \geq |v^1_i - v^2_i| - |v^1_i - \tilde{v}^1_i| - |v^2_i - \tilde{v}^2_i|\\
& \geq 100 \cdot \alpha / n - 10 \cdot \alpha / n - 10 \cdot \alpha / n\\
& = 80 \cdot \alpha / n.
\end{align*}
which contradicts item 1. The lemma follows.
\end{proof}

\lbadditive*

\begin{proof}
We instantiate our construction using the following parameters.
The number of edges from a vertex to its $i$th child is $(1/\epsilon)^i$ ($p = 1/\epsilon$).
Each vertex of $H$ has $t = 1/2 \cdot \log_{p} n$ children.
The tree $H$ has depth $d = \frac{\log (101\alpha)}{2 \log (1-\epsilon)} - 2\geq 1$.
Note that $d = \Theta(\log (1/\alpha))$.
Finally, $S_0, S_1$ have $s = n/4$ leaves.

Let us now bound the size of the graph.
The number of leaves of $H$ is
\begin{equation}\label{eq:leaves}
t^d = (1/2 \cdot \log_{1/\epsilon} n)^d = \left(\frac{\log n}{ 2 \log {1/\epsilon}}\right)^{\Theta(\log (1/\alpha))} = \log^{\Theta(\log (1/\alpha))} n = (1/\alpha)^{\Theta(\log \log n)},
\end{equation}
where in the third step we use the fact that $\frac{\log n}{ 2 \log {1/\epsilon}} = \log^{\Theta(1)} n$ for sufficiently large $n$.

$R, S_0$ and $S_1$ have $n/4$ vertices each.
By \cref{eq:leaves} and the assumption on $\alpha$, $H$ has $o(n)$ vertices, and so with the additional isolated vertices, the graph has exactly $n$ vertices.

The number of edges incident to $R, S_0$ and $S_1$ is $O(n)$.
The number of children of each internal vertex of $H$ is 
\[
\Theta(p^t) = \Theta(p^{1/2 \cdot \log_p n}) = \Theta(n^{1/2}).
\]
Thus, the total number of edges in $H$ is $(1/\alpha)^{\Theta(\log \log n)}\cdot n^{1/2} = n^{o(1)} \cdot \Theta(n^{1/2})$.
Hence, we conclude that the graph has $n$ vertices and $O(n)$ edges.

Observe that as we add edges, leaves $\ell_1, \ell_2, \ldots$ become reachable from $R$ exactly in the order of their indices.
Fix any leaf $\ell_j$ of $H$.
Denote by $\pi^b$ and $\pi^a$, respectively, the PageRank vectors just before $\ell_j$ is reachable from $R$ and just after all edges on the path from $R$ to $\ell_j$ are added.

We use the interpretation of PageRank based on probability mass.
Before $\ell_j$ is reachable from $R$, it may receive probability mass only from its ancestors in $H$.
Hence,
\[
\pi^b_{\ell_j} \leq (d+1)/n = \Theta(\log(1/\alpha)) / n = o(\log n) / n.
\]

Moreover, since PageRank is a $\epsilon$ fraction of the probability mass entering each vertex, by \cref{lem:pagerank-leaf},
\[
\pi^a_{\ell_j} \geq \epsilon \cdot (1-\epsilon)^{2d+2}/4 = \epsilon \cdot (1-\epsilon)^{\frac{\log (101\alpha)}{\log (1-\epsilon)}-2}/4 = 101\cdot \epsilon \cdot \alpha \cdot (1-\epsilon)^{-2}.
\]

The increase to PageRank of $\ell_j$ is thus at least  $\pi^a_{\ell_j} - \pi^b_{\ell_j} \geq 100\cdot \epsilon \cdot \alpha (1-\epsilon)^{-2}$.
Hence, after the insertion there is at least $100 \alpha (1-\epsilon)^{-2} / 4$
"new" probability mass at $\ell_j$.
Since every two hop path from $j$ leads to a leaf in $S_{j \bmod 2}$, each of these leaves will receive a least
\[
100 \cdot \alpha \cdot (1-\epsilon)^{-2} / 4 \cdot (1-\epsilon)^2 / s = 100 \cdot \alpha / (4s).
\]
new probability mass (since only $(1-\epsilon)$  fraction of the probability mass is transferred along each hop).
By \cref{obs:loop} all of that probability mass ends up increasing the PageRank of the leaf.
Therefore the PageRank of each of these $s$ leaves increases by $100 \cdot \alpha / (4s) = 100 \cdot \alpha / (4 \cdot n / 4) = 100 \cdot \alpha / n$.

We now use \cref{lem:coordinates} with $v_1 = \pi^b$, $v_2 = \pi^a$ and $\tilde{v}_1$ and $\tilde{v}_2$ being any PageRank vectors giving $\alpha$-additive approximation and infer that $\Omega(n)$ coordinates of any approximate PageRank vector must be updated in order to maintain $\alpha$-additive approximation.
This happens for each leaf of $H$, and so by \cref{eq:leaves} the Lemma follows.
\end{proof}

\subsection{Lower Bound for Maintaining Multiplicative Approximation}
\label{sec:LB-multiplicative-approx}

\thmmub*

\begin{proof}
We instantiate our construction as follows.
Each non-leaf vertex $v$ of $H$ has exactly $t =  \delta/2 \log n/ \log \log n$ children, with  $(\log^2 n)^i$ parallel directed edges from $v$ to the $i$-th child of $v$ ($p = \log^2 n)$. It follows that the total outdegree of each internal vertex in $H$ is $O(n^\delta)$. 
The depth of $H$ is set to be $d = \log_t(n^{1-2\delta}) = \Theta(\log n / \log \log n)$, so that $H$ has $n^{1-2\delta}$ vertices, and the total number of edges in $H$ is $O(n^{1-\delta})$. 
Finally, both $S_0$ and $S_1$ have $s = n^{1-2 \delta}$ vertices.

Fix a leaf $\ell_j$ of $H$ and consider the state of the algorithm right after all on the path from the root of $H$ to $\ell_j$ have been added.
By \cref{lem:pagerank-leaf}, the probability mass entering $\ell_j$ is at least.
\[
 (1-\epsilon)^{2d+2}/4 =  (1-\epsilon)^{\Theta(\log n / \log \log n)} = n^{\Theta(-1/\log \log n)}.
\]

Out of this probability mass a constant fraction reaches the leaves of $S_{j \bmod 2}$.
In particular, the PageRank of each such leaf is at least $\epsilon \cdot n^{\Theta(-1/\log \log n)} / n^{1-2\delta} \geq n^{-\delta}$.

Moreover, out of the probability mass from $R$ the fraction that reaches $\ell_j$ is at least
\[
(1-1/p)^d = (1-1/\log^2 n)^{\Theta(\log n / \log \log n)} \geq 1-1/\log n.
\]
out of all probability mass that reaches the leaves of $H$ from $R$.
Observe that compared to this probability mass (which is a constant), the total probability mass generated by all vertices of $H$ is negligible.
As a result, the ratio of probability mass that reaches $S_{j \bmod 2}$ to the probability mass that reaches $S_{(j+1) \bmod 2}$ is
\[
\frac{1 - 1/\log n}{1 / \log n} = \Theta(\log n).
\]
This implies that when we add all edges on a path from $R$ to $\ell_j$, the PageRanks of leaves of $S_{j \bmod 2}$ increase by a factor of $\Theta(\log n)$ and so the PageRank estimates of all these $\Omega(n^{1-2\delta})$ vertices must be changed. 
Since a total of $m = O(n)$ edges are added, and since this occurs once for each of the $\Omega(n^{1-\delta})$ leaf vertices in $H$, we obtain a total of $\Omega(n^{2-3\delta})$ PageRank estimate updates, which is the desired result after rescaling $\delta$ by a constant.
\end{proof}

\subsection{Lower bound for the ForwardPush algorithm}

\begin{theorem}\label{thm:fp}
Consider running the ForwardPush~\cite{zhang2016approximate} algorithm whose error parameter is set to ensure that the algorithm maintains additive $\alpha$ approximation of PageRank.
For any $\delta > 0$, each sufficiently large $n \geq 1$ and $\epsilon \in (0.01, 0.99)$ there exists a graph  on $n$ vertices and a sequence of $O(n)$ edge insertions, such that the algorithm runs in $\Omega(n^{2-\delta})$ time.
\end{theorem}

\begin{proof}
We use our construction with the same settings as in the proof of \cref{thm:lb-explicit}.
Specifically, $t =  \delta/2 \log n/ \log \log n$, $p = \log^2 n$, $d = \log_t(n^{1-2\delta}) = \Theta(\log n / \log \log n)$ and $s = n^{1-2 \delta}$.

The ForwardPush algorithm can be explained using the probability mass interpretation.
The algorithm maintains a \emph{residual} on each vertex $u$, denoted by $R_u$.
This residual can be positive or negative.
Initially, the residual of each vertex is $1/n$.

The residual is a probability mass that still has to be pushed to the neighbors of $u$.
The algorithm maintains two invariants
\begin{enumerate}
    \item $|R_u| \le \gamma \deg(u)$ for each vertex $u \in V$, where $\gamma$ is an accuracy parameter.
    \item If we keep pushing the residuals, the PageRank estimates converge to the exact PageRank values.
\end{enumerate}

For any vertex $u$ that violates the invariant, that is satisfies $|R_u| / \deg(u) > \gamma$, the algorithm executes a \emph{push} operation, which takes time $\Theta(\deg(u))$ and pushes a $1-\epsilon$ fraction of the residual to the outneighbors of $u$ and uses a $\epsilon$ fraction of the residual to increase the PageRank of $u$.
The residual of $u$ is then set to $0$.
Upon an insertion of an edge $uv$, the algorithm decreases $R_u$ by $\Delta = \Theta(\pi_u) / \deg(u)$ and increases $R_v$ by $\Delta$.
Then, it restores the invariant by executing push operations.

In the following part of the proof we use the following observation, which follows from the second algorithm invariant.

\begin{observation}\label{obs:residuals}
Fix a vertex $v$ and denote by $D_v$ the set of vertices that have a directed path to $v$. We assume $v \in D_v$.
Then, the total additive error of the PageRank estimate maintained by the ForwardPush algorithm is at most $\sum_{u \in D_v} |R_u|$.
\end{observation}

By using the second algorithm invariant, we get that ForwardPush ensures that the total additive error is $\sum_{u \in V} |R_u| \le \sum_{u \in V} \gamma \deg(u) = \gamma m$. Therefore, to ensure an additive $\alpha$ approximation of PageRank, we set $\gamma m \le \alpha$, implying $\gamma \le \alpha / m$.
We note that it is easy to come up with an example where this analysis is tight up to a constant factor.

We now analyze ForwardPush algorithm on our hard instance.
Since the number of edges in our graph is  $\Theta(n)$, we invoke ForwardPush with the approximation parameter $\gamma = \Theta(1/n)$. 
We claim that with this value of $\gamma$, the residual values are propagated often enough so that over $\Theta(n)$ edge insertions described above, ForwardPush makes $\Omega(n^{2 - \delta})$ updates.

We use the observations from the proof of \cref{thm:lb-explicit} that the PageRank of a vertex $c_i$ ($i \in \{0, 1\}$) is $n^{\Theta(-1/ \log \log n)}$ and, as we add edges, increases by a $\Theta(\log n)$ factor each time we fully add a path from $R$ to a leaf $\ell_j$, such that $i = j \bmod 2$.

We now use \cref{obs:residuals} to show that the ForwardPush maintains a constant factor approximate of the PageRank estimates of $c_0$ and $c_1$.
Indeed, these vertices can only be reached from $R$, $H$ or from themselves.
We now bound the residuals of these vertices.
The residuals of the vertices of $R$ are set to $0$ the moment each of these vertices performs the first push operation and are then never updated.
The residual of each vertex $v$ of $H$ satisfies
$|R_v| / \deg(v) \leq \alpha / m$ which implies $|R_v| \leq \Theta(\deg(v)) / m = \Theta(n^{\delta-1})$.
Finally, the residual of $c_0$ (and, similarly $c_1$) satisfies
$|R_{c_0}| / \deg(c_0) \leq \alpha / m$, which gives
$|R_{c_0}| \leq \Theta(n^{1-2\delta}) / m = \Theta(n^{-2\delta})$.
By applying \cref{obs:residuals} we have that the additive error the PageRank estimates of $c_0$ and $c_1$ is at most
\[
\Theta(n^{\delta-1}) \cdot \Theta(n^{1-2\delta}) + \Theta(n^{-2\delta}) = \Theta(n^{-\delta}).
\]
These additive errors are negligible comared to the PageRanks of these vertices, which is $n^{\Theta(-1 / \log \log n)}$.
Hence, the algorithm maintains constant-factor estimates of the PageRanks of $c_0$ and $c_1$.
As a result, when the exact PageRank values change by a factor of $\Theta(\log n)$, the algorithm updates their estimates.
However, the ForwardPush algorithm only updates a PageRank estimate of a vertex $u$ when either it executes a push operation on $u$ or adds an outgoing edge from $u$.
Since all outgoing edges of $c_0$ and $c_1$ have been added in the beginning, we get that the algorithm executes a push operation on $c_0$ for half of leaves of $H$.
Each such operation takes $\Theta(\deg(c_0)) = \Theta(n^{1-2\delta})$ time and so the overall running time of the algorithm is $\Theta(n^{1-2\delta} \cdot n^{1-2\delta}) = \Theta(n^{2-4\delta})$ which, after tweaking $\delta$ by a constant factor, gives the desired.
\end{proof}

\section{Approximating PageRank by Maintaining Dynamic Random Walks}
\label{sec:random-walks}
In this section we review the algorithm for approximating PageRank by maintaining random walks.
This algorithm is a dynamic version of \cref{alg:PageRank-approximation} and has been previously described by Bahmani et al.~\cite{bahmani2010fast}.
We provide a detailed proof of correctness of the algorithm, which to the best of our knowledge has not been included in any prior work.

The algorithm relies on maintaining $O_\eps(n \log n)$ random walks and re-sampling their parts as necessary.
In this section, we present data structures that we use to efficiently maintain and re-sample those random walks.
\cref{sec:appendix-edge-insertions-undirected} presents our approach on an edge insertion, while \cref{sec:appendix-edge-removals} describes how our algorithms handle edge deletions.
We being by describing the problem setup.

\textbf{Setup.}
Following \cref{prop:walksEstimate}, to approximate the PageRank it suffices to sample $R = O(\log n / (\eps \alpha^2))$ PageRank walks from each vertex. A PageRank walk is a random walk $w$, whose length $\ell_w$ is sampled from geometric distribution with parameter $1 - \eps$.
Even though a given walk may get re-routed after edge insertions or deletions, it is crucial that the the length of each walk remains \textbf{fixed} throughout the entire execution of the algorithm.
Otherwise, it is easy to construct examples where the lengths of the maintained walks no longer follow the right distribution.

We maintain two types of data structures. For each each vertex $v$ and $t = 0 \ldots O(\log n / \eps)$, we maintain a binary search tree $S_{v, t}$ which stores all the walks whose $t$-th vertex if $v$. For each edge $e$, we maintain the binary search tree $W_e$ consisting of the walks passing through $e$.

\subsection{Edge Insertion}
\label{sec:appendix-edge-insertions-undirected}

When an edge $(u, v)$ is inserted, we re-sample some of the walks passing through $u$. This re-sampling is done by first performing rejection sampling on each walk and, second, by choosing an appropriate position where each of the rejected walks should be re-sampled. Choosing an appropriate position from where to re-sample $w$ is trivial in case when $w$ passes through $u$ once. However, it might be the case that $w$ passes through that vertex multiple times, and a more careful consideration is required.
At a high level, we iterate through all segments of $w$ and for each segment of $w$ that leaves $u$ we toss a coin.
Then, with probability $1/d_u$, where $d_u$ is the degree of $u$ after the update, we reroute $w$ starting from the considered segment, and terminate the update procedure for $w$.

Each walk has a unique ID associated with it. These IDs are integers ranging from $1$ through the number of walks we maintain.
Each vertex and each edge keeps track of which walks are passing through them.

Given a vertex $v$ and integers $i$ and $t$, it will be convenient to be able to sample the $i$-th walk whose $t$-th vertex is $v$. It will become clear why such operation is needed when we describe how to handle edge insertions. 
To be able to implement this operation efficiently, we store the IDs of walks whose $t$-th vertex is $v$ in a binary tree; we use $S_{v, t}$ to refer to this binary tree. 
Then, the $i$-th walk can be easily fetched via a search within that tree. 
The maximum value of $t$ to consider is upper-bounded by the maximum length of the walks.

Assume that we insert an edge $e = (u, v)$. Let $d_u$ be the out-degree of $u$ \emph{after} adding $e$. Consider a walk $w$ that at some point got to $u$ and continued to $u$'s neighbors. If $e$ was present in the graph at that point, with probability $1 / d_u$ the walk $w$ would have continued along $e$, and with probability $1 - 1/d_u$ the walk $w$ would have chosen some other neighbor of $u$. However, $w$ was sampled before $e$ was in the graph, and our aim now is to correct this distribution and account for the insertion of $e$. The idea is to use rejection sampling, which we provide as \cref{alg:walk-insertion}.

\begin{algorithm}
\caption{A procedure executed after edge $e = (u, v)$ is inserted.}
\label{alg:walk-insertion}
\begin{algorithmic}[1]
\STATE{$W \gets \emptyset$}
\STATE{Let $\ell$ be the length of longest generated walk.}
\FOR{$t = 1 \ldots \ell$ \label{line:loop-over-lengths}}

    \STATE{Sample each walk from $S_{u, t}$ with probability $1/d_u$ in the following way.
    First, \\select an integer $r_{u, t}$ from the binomial distribution with parameters $|S_{u, t}|$ and\\ $1/d_u$. Second, select $r_{u, t}$ integers uniformly at random and without repetition\\ from $[1, |S_{u, t}|]$. Then, for each of those integers $i$ select the $i$-th walk from $S_{u, t}$.\\ If $e$ is an undirected edge, apply the same steps for $S_{v, t}$. \label{line:select-walks-for-resampling}
    }
    \STATE{For each walk $w$ selected in the last step such that $w \notin W$, add $w$ to $W$ and \\label $w$ by $t$.\label{line:add-w-to-W}}
\ENDFOR

\FOR{each $w \in W$}
    \STATE{Let $j$ be the label remembered for $w$ on Line~\ref{line:add-w-to-W}.}
    \STATE{
        Generate walk $w'$ with the following properties:
    \begin{itemize}
    	\item The walks $w$ and $w'$ have the same length.
    	\item The vertex-prefixes of length $j$ of $w$ and $w'$ are the same.
    	\item After that prefix, if $w$ has more than $j$ vertices, $w'$ walks along $e$.
    	\item The remaining edges of $w'$ are chosen randomly, i.e., the rest of $w'$ is a newly generated random walk.
    \end{itemize}
    }
    
    \STATE{Update the data structures by removing $w$ and inserting $w'$.}
\ENDFOR

\end{algorithmic}
\end{algorithm}

The for-loop on Line~\ref{line:loop-over-lengths} of \cref{alg:walk-insertion} is in an efficient way of selecting walks passing through $u$ and $v$ that need to be re-sampled. 
Since the length of each walk follows a geometirc distribution with parameter $1-\epsilon$, it is easy to see that with high probability the walks have length $O(\log n / \eps)$, and hence $\ell \in O(\log n / \eps)$.

\emph{Remark:}
To the best of our understanding, on an insertion of an edge $(u,v)$, the prior work \cite{zhan2019fast} re-samples a walk passing through $u$ from the first occurrence of $u$ in the walk, if there is any such occurrence (for details, see \cite{zhan2019fast}). Such re-sampling does not account for the case when a walk passes through $u$ multiple times and leads to biases in randomness.

\subsection{Edge Deletions}
\label{sec:appendix-edge-removals}
\cref{alg:walk-removals} presents our procedure executed after deleting an edge.

\begin{algorithm}[t]
\begin{algorithmic}[1]
\STATE{Let $W_e \subseteq W$ be the list of walks passing through $e$.}
\FOR{$w \in W_e$}
	\STATE{Let $w_p$ be the longest prefix of $w$ not containing  $e$.}
	\STATE{Let $w'$ be a walk of length $|w|$ such that $w'$ has $w_p$ as its prefix, and the\\ remainder of $w'$ is a random walk.}
	\STATE{To update $W$, remove $w$ from $W$ and the corresponding data structures, and \\insert $w'$.}
\ENDFOR
\end{algorithmic}
\caption{A procedure executed after edge $e$ is deleted.}
\label{alg:walk-removals}
\end{algorithm}

Let $e$ be a deleted edge, and let $W_e \subseteq W$ be the list of walks passing through $e$. Clearly each $w \in W_e$ needs to be rerouted. 
The following lemma states that $W$ updated by executing \cref{alg:walk-removals} is a set of independent random walks.

\begin{lemma}
\label{lemma:removals-maintain-random-walks}
	Let $W$ be the set of walks that our algorithm maintains. Assume that $e$ gets deleted, and let $W'$ be the updated list of walks as described in \cref{alg:walk-removals}. 
	If $W$ consists of random walks sampled independently, then $W'$ is also a set of random walks sampled independently.
\end{lemma}
\begin{proof}
	The edges of walks throughout the algorithm are sampled independently of each other, so walks are independent by construction. 
	We focus on showing how deletion of an edge affects randomness of a single walk.

	Consider a walk $w \in W$ originating at vertex $w_1$. Let $w_i$ be the $i$-th vertex of $w$, $w_{1 \ldots i}$ be the prefix of length $i$ of $w$, and $k$ be the length of $w$.
	Walk $w$ is random iff for each $i \ge 2$ and each $u \in N(w_{i - 1})$ it holds
	\begin{equation}\label{eq:del-assumption}
		\prob{w_i = u \ |\ w_{1 \ldots i-1}} = \frac{1}{d(w_{i - 1})}.
	\end{equation}
	Let $w'$ be the updated walk $w$, $d'(v)$ be the updated degree of vertices after $e$ gets deleted and $u'$ be a neighbor of $w_{i - 1}'$ after deletion of $e$. 
	Note: we are \textbf{not} assuming that $w$ contains $e$, so it might be the case that $w = w'$.
	We want to show that $\prob{w_i' = u' \ |\ w_{1 \ldots i-1}'} = 1 / d'(w_{i - 1}')$. 
	We have
	\begin{align}
		& \prob{w_i' = u' \ |\ w_{1 \ldots i - 1}'} \label{eq:random-walks-del} \\
		= &  \prob{w_i' = u' \ |\ w_{1 \ldots i - 1}', e \in w_{1 \ldots i}} \cdot \prob{e \in w_{1 \ldots i}} \label{eq:random-walks-del-term-e-in} \\
		& + \prob{w_i' = u' \ |\ w_{1 \ldots i - 1}', e \notin w_{1 \ldots i}} \cdot \prob{e \notin w_{1 \ldots i}}. \label{eq:random-walks-del-term-e-notin}
	\end{align}
	
	\textbf{Analyzing \eqref{eq:random-walks-del-term-e-in}.}
	We first handle \eqref{eq:random-walks-del-term-e-in}. Recall that $w'$ is constructed by keeping only the prefix of $w$ up to the first occurrence of $e$, and the rest of the walk of $w'$ is random and independent of any other state of the algorithm (see \cref{alg:walk-removals}). Hence, we have
	\[
		\prob{w_i' = u' \ |\ w_{1 \ldots i - 1}', e \in w_{1 \ldots i}} = \frac{1}{d'(w_{i - 1}')}.
	\]

	\textbf{Analyzing \eqref{eq:random-walks-del-term-e-notin}.}
	Now consider term \eqref{eq:random-walks-del-term-e-notin}. If $w_{1 \ldots i}$ does not contain $e$, then $w_{1 \ldots i}' = w_{1 \ldots i}$ and we have
	\begin{align*}
		& \prob{w_i' = u' \ |\ w_{1 \ldots i - 1}', e \notin w_{1 \ldots i}} \\
		= & \prob{w_i = u' \ |\ w_{1 \ldots i - 1}, e \notin w_{1 \ldots i - 1}, e \neq \{w_{i - 1}, w_i\}}.
	\end{align*}

	There are two cases:
	\begin{enumerate}[(a)]
		\item Case $w_{i - 1} \notin e$: from \eqref{eq:del-assumption} we have
	\begin{align*}
		& \prob{w_i = u' \ |\ w_{1 \ldots i - 1}, e \notin w_{1 \ldots i - 1}, e \neq \{w_{i - 1}, w_i\}, w_{i - 1} \notin e} \\
		= & \prob{w_i = u' \ |\ w_{1 \ldots i - 1}, e \neq \{w_{i - 1}, w_i\}, w_{i - 1} \notin e}  \\
		= & \frac{1}{d(w_{i - 1})} = \frac{1}{d'(w_{i - 1})} = \frac{1}{d'(w_{i - 1}')}.
	\end{align*}
	In the last chain of equalities we used that once we condition on $w_{1 \ldots i - 1}$, then \eqref{eq:del-assumption} is a function of only $w_{i - 1}$ and not on any other content of $w_{1 \ldots i-1}$, e.g., whether $e \in w_{1 \ldots i - 1}$ or not.
	\\
	\textbf{Note}: The choice of $e$ is independent of our data structures and the randomness the algorithm uses. However, in the case of non-oblivious adversary, i.e., in case of the adversary who sees the state of our algorithm, the updated edge $e$ could be chosen based on the randomness used to generate $w$, and hence the above sequence of equalities would not hold.
	
	\item Case $w_{i - 1} \in e$: we have the following
		\begin{align*}
			& \prob{w_i = u' \ |\ w_{1 \ldots i - 1}, e \notin w_{1 \ldots i - 1}, e \neq \{w_{i - 1}, w_i\}, w_{i - 1} \in e} \\
			= & \frac{\prob{w_i = u' \wedge e \neq \{w_{i - 1}, w_i\} \ |\ w_{1 \ldots i - 1}, e \notin w_{1 \ldots i - 1}, w_{i - 1} \in e}}{\prob{e \neq \{w_{i - 1}, w_i\} \ |\ w_{1 \ldots i - 1}, e \notin w_{1 \ldots i - 1}, w_{i - 1} \in e}} \\
			= & \frac{1 / d(w_{i - 1})}{(d(w_{i - 1}) - 1) / d(w_{i - 1})} = \frac{1}{d(w_{i - 1}) - 1} = \frac{1}{d'(w_{i - 1}')}.
		\end{align*}
	\end{enumerate}
	
	\textbf{Showing \eqref{eq:del-assumption} for $w'$.}
		The analysis of \eqref{eq:random-walks-del-term-e-in} and \eqref{eq:random-walks-del-term-e-notin} together with \eqref{eq:random-walks-del} implies
		\begin{align*}
            & \prob{w_i' = u' \ |\ w_{1 \ldots i - 1}'} \\
            = & \frac{1}{d'(w_{i - 1}')} \cdot \prob{e \in w_{1 \ldots i}} + \frac{1}{d'(w_{i - 1}')} \cdot \prob{e \notin w_{1 \ldots i}} \\
            = & \frac{1}{d'(w_{i - 1}')}.
		\end{align*}
			
\end{proof}

\subsubsection{Re-sampling Walks from Scratch}\label{sec:introRemark}

We now give a simple example that shows why re-sampling affected walks from scratch after a deletion would not properly maintain random walks.
We note that this approach was suggested as a valid alternative by Bahmani et al.~\cite{bahmani2010fast}.

Consider a path graph on $5$ vertices; let the graph be $1-2-3-4-5$. Consider a random walk $w$ of length $2$ originating at vertex $3$ and visiting vertices $w_1, w_2, w_3$, i.e., $w_1 = 3$. Next, a deletion of $e = \{4, 5\}$ occurs. Let $w'$ be obtained from $w$ as follows: if $w$ contains $e$, then $w'$ is a new random walk of length $2$ originating at $3$; otherwise, $w'$ equals $w$. Now, if we denote the vertices on $w'$ by $w_1', w_2', w_3'$, we have
\begin{align*}
	\prob{w_2' = 4} = & \prob{w_2' = 4\ |\ \{4, 5\} \notin w} \prob{\{4, 5\} \notin w} \\
	& + \prob{w_2' = 4\ |\ \{4, 5\} \in w} \prob{\{4, 5\} \in w} \\
	= & \prob{w_2 = 4\ |\ \{4, 5\} \notin w} \prob{\{4, 5\} \notin w} \\
	& + \prob{w_2' = 4\ |\ \{4, 5\} \in w} \prob{\{4, 5\} \in w} \\
	= & \prob{w_2 = 4 \text{ and } \{4, 5\} \notin w} + \frac{1}{2} \cdot \frac{1}{4} \\
	= & \frac{1}{4} + \frac{1}{8}.
\end{align*}
However, for $w'$ to be random it should hold $\prob{w_2' = 4} = 1/2$.

\section{Near-Optimal Additive Approximation Algorithm}\label{sec:improved-bounds}

In this section, we analyze the algorithm from \cref{sec:random-walks} in the context of  dynamically maintaining \emph{additive} approximation of PageRank.
Namely, we show that when considering the incremental or decremental setting for directed graphs, an $\alpha$ additive PageRank approximation can be maintained in $(1/\alpha)^{O_\eps(\log \log n)}$ amortized update time, even for an adversarially chosen graph and a sequence of edge updates.
Perhaps surprisingly, \cref{thm:LB-additive} shows that, for a constant $\eps$, this running time complexity is essentially tight.

\directedadditiveup*

\begin{comment}

This result is based on a slightly modified algorithm of Bahmani et al.~\cite{bahmani2010fast}, and our contribution lies in providing the new analysis.
The algorithm of Bahmani et.~al~\cite{bahmani2010fast} maintains a collection of random walks from each vertex, and is essentially a dynamic version of  Algorithm~\ref{alg:PageRank-approximation}.
At a high level, to ensure that the maintained walks have the right distribution, the algorithm has to (a) upon an edge deletion, resample (parts of) any random walks that used the deleted edge, and (b) upon an insertion of an edge $xy$ reroute any random walk visiting the vertex $x$ to use the new edge with probability $1 / \deg(x)$.
We describe this algorithm, including some subtle details that we skipped here for simplicity, in \cref{sec:appendix-details-missing-from-improved-bounds}.
The total update time of the algorithm is proportional to the number of random walks that need to be updated in each step.
\end{comment}

Our new analysis is based on two ideas.
First, we show that if we \emph{limit the lengths} of walks in \cref{alg:PageRank-approximation} to a constant, we obtain a constant additive approximation of the PageRank vector.
This is thanks to the fact that a constant fraction of all walks have length $O(1/\epsilon)$, and so this truncation only affects a constant factor of the walks.

\begin{lemma}
\label{lemma:constant-length}
    Let $\pi$ be the PageRank of a directed graph $G$. Then, with high probability, \cref{alg:PageRank-approximation} for $\ell = \lceil 2/\eps \cdot \log{(2/(\alpha \eps))} \rceil$ outputs a vector $\piadd$ such that
    $\|\pi - \piadd\|_1 \le  5 \frac{\alpha}{1 - \eps}$.
\end{lemma}
To keep the flow of high-level ideas uninterrupted, the proof of \cref{lemma:constant-length} is given in \cref{sec:proof-of-constant-length}.

The second idea is an observation which bounds the maximum number of times a walk can be affected by adding edges (edge deletions can use a symmetric argument).
To explain the idea let us see what happens when we want to maintain a random outgoing edge $e$ of a vertex undergoing insertions of outgoing edges.
Clearly when we insert the $d$-th outgoing edge we need to update $e$ to be equal to $d$ with probability $1/d$.
By a harmonic sum argument, the expected number of times $e$ needs to be updated in the course if $k$ insertions is only $O(\log k)$.
We generalize this argument to walks of length $\ell$ as follows.

\begin{lemma}\label{lem:reroute}
Let $G$ be a directed graph undergoing edge insertions (or deletions). The 
total number of times a random walk of length $\ell$ is being regenerated
is bounded by $O(\log^{\ell}n)$ in expectation. 
\end{lemma}
\begin{proof}
We are going to prove this bound by induction, i.e., let us denote by $f(i)$ the upper bound on expected number of times the walk of length $i$ is regenerated. Consider a random walk $w$ of length $1$ starting in a vertex $v$. Consider insertion of an edge incident to $v$. The probability that $w$ is regenerated at this moment is $1/d_v$. As we consider incremental setting the expected number of times $w$ is regenerated is bounded by 
\[
f(1) = \sum_{i=1}^n \frac{1}{i} \le \ln n. 
\]
Now consider a walk $w$ of length $\ell$ starting at $v$. Similarly as above we can bound the number of changes to $w$ as
\[
f(\ell) = \sum_{i=1}^n \frac{1}{i} \cdot f(\ell-1)\le \ln n \cdot f(\ell-1)=\ln^{\ell}n,
\]
what finishes the proof. Symmetric argument can be applied in the decremental case.
\end{proof}

The above lemma implies that for $\ell = \lceil 2/\eps \cdot \log{(2/(\alpha \eps))} \rceil$ the amortized cost of maintaining each walk is $(1/\alpha)^{O(\log \log n)}$ for a constant $\eps$. 
As we generate $O(n\log n)$ walks in \cref{alg:PageRank-approximation} the total cost of maintaining $5\alpha/(1-\eps)$-approximation in incremental or decremental setting is $O(m + n \cdot (1/\alpha)^{O(\log \log n)})$.

\subsection{Proof of \cref{lemma:constant-length}}
\label{sec:proof-of-constant-length}
%\begin{proof}
Define $\hll = \lceil 2/\eps \cdot \log{(2/(\alpha \eps))} \rceil$. 
Let $\tpi$ be the output of \cref{alg:PageRank-approximation} for $\ell = \infty$, and $\piadd$ the output for $\ell = \hll$. As discussed, it is known, e.g., see \cite{bahmani2010fast,lkacki2020walking}, that $|\pi_v - \tpi_v| \le \alpha \pi_v$. As $\sum_v \pi_v = 1$, this further implies $\|\pi - \tpi\|_1 \le \alpha$.

Next, we compare $\piadd$ and $\tpi$. Difference between these two vectors can be expressed by the following two quantities: (1) $|W|$, which in turn affects the scaling on Line~\ref{line:tpi-assign}; and (2) the value of $\bX_v$, which affects the numerator on Line~\ref{line:tpi-assign}. We analyze both of these quantities.

\textbf{Analysis for $|W|$:}
For $\ell = \hll$, a walk has length at most $\hll$ with probability $\eps \sum_{j = 0}^{\hll} (1-\eps)^j = 1 - (1-\eps)^{\hll + 1} \ge 1 - \eps/2$, where we used that $1-x \le e^{-x}$ for $x \in [0, 1/2]$. Hence, $\ex{|W|} \ge n R (1-\eps/2)$. 
By using a Chernoff bound we can prove that with high probability it holds $|W| \ge n R (1 - \eps)$. The proof proceeds as follows. In the summation above, there are only $\ell$ different values of $j$ that affect $\ex{|W|}$. For a fixed $j$, the contribution to $|W|$ can be expressed as a sum of \textbf{independent} $0/1$ random variables -- a random variable per each of the $nR$ walks, denoting whether the given walk has length length $j$ or not. Hence, for a fixed $j$ we apply the Chernoff bound to show it concentrates well, and then by the union bound over all $\ell$ values of $j$ we get the desired concentration for $|W|$.

\textbf{Analysis for $\bX_v$:}
By definition, $\piadd$ only accounts for the contribution to $\bX_v$ by the appearances of $v$ which are within walks of length at most $\ell$; $\bX_v$ is defined in \cref{alg:PageRank-approximation}. Let $\bY_v$ be the appearances of $v$ for which $\piadd$ does not but $\tpi$ does account for.

Now, we upper-bound $\sum_v \bY_v$:
\begin{align*}
    \ex{\sum_v \bY_v} = & nR (1-\eps)^{\hll+1} \cdot (\hll+1) + \sum_{j = \hll + 2}^{\infty} n R (1-\eps)^j \\
    \le & 2 nR \alpha  + nR (1 - \eps)^{\hll + 2} \sum_{j = 0}^{\infty}(1-\eps)^j \\
    =& 2 nR \alpha + \frac{nR}{\eps}(1 - \eps)^{\ell + 2} \\
    \le & 2 nR \alpha + nR \eps \alpha^2 / 4 \\
    \le & 3 nR \alpha.
\end{align*}
In the derivation above, we used $(1-\eps)^{\hll + 1}(\hll + 1) \le (\alpha \eps / 2)^2 (\hll + 1) \le (\alpha \eps / 2)^2 2 \hll \le 2\alpha$.
To prove that $\sum_v \bY_v \le 4 nR \alpha$ with high probability, it suffices to proceed the same way as for our analysis of $\ex{|W|}$. In the analysis, we need the observation that $\sum_{j >  c \log n / \eps} nR(1-\eps)^j < 1/n$ for a sufficiently large constant $c$. In other words, there are only $O(\log n)$ different values of $j$ that substantially contribute to $\sum_v \bY_v$ and over which is needed to take the union bound.

Our analysis now implies that additive approximation of \cref{alg:PageRank-approximation} for $\ell = \hll$ is with high probability upper-bounded by $\tfrac{\alpha}{1-\eps} + 4 \tfrac{\alpha \eps}{1 - \eps} \le 5 \tfrac{\alpha}{1-\eps}$. The first term is coming from the fact that $\piadd$ is computed by rescaling $\bX_v$ by $|W|/\eps \ge (1-\eps) nR/\eps$ as opposed to rescaling by $nR/\eps$, as it is done when computing $\tpi$. The second term is coming from the fact that the loss between $\tpi$ and $\piadd$ in the numerator of Line~\ref{line:tpi-assign} is at most $4 nR \alpha$ with high probability, which is divided by $|W|/\eps \ge n R(1-\eps)/\eps$.
%\end{proof}

\section{Efficient Multiplicative Approximation in Undirected Graphs}
\label{sec:upper-bound-undirected}

In this section, we describe how to maintain approximate PageRank of undirected graphs under edge deletions and insertions even if the goal is to maintain a multiplicative approximation.
Our approach takes $\polylog{n}$ time per update and is also based on the algorithm from \cref{sec:random-walks}.

\thmundirected*

Our analysis relies on the following (folklore) claim, which states that the number of the walks passing through an edge is fairly small.
\begin{lemma}[Folklore]\label{lem:load}
    Let $G$ be an undirected graph.
	Consider a set of random walks $W$ of length $\ell < n$ each, such that there are $d_v$ walks originating at vertex $v$. Then, with high probability an edge $e$ is contained in $O(\ell \cdot \log{n})$ of those walks.
\end{lemma}
\begin{proof}
	Observe that the number of walks in $W$ originating at each vertex $v$ is proportional to the stationary distribution of $v$. Hence, the number of walks of $W$ whose $i$-th vertex is $v$ in expectation equals $d_v$, for each $1 \le i \le \ell$. Therefore, the number of walks of $W$ whose $i$-th edge is $e = \{u, v\}$ (either as $u \rightarrow v$ or $v \rightarrow v$) in expectation equals $2$, for each $1 \le i \le \ell$.
	
	Let $X_{e, i}$ be the number of walks whose $i$-th edge equals $i$. From our discussion, $\ex{X_{e, i}} = 2$. Also, $X_{e, i}$ is a sum of $0/1$ independent random variables $Y_{v, j, i}$, where $Y_{v, j, i}$ means that the $i$-th edge of the $j$-th walk originating at $v$ equals $e$. Hence, by applying the Chernoff bound, we obtain that with high probability it holds that $X_{e, i} \in O(\log n)$. By taking the union bound over all $1 \le i \le \ell$ and over all the vertices, we prove the desired claim.

\end{proof}
As a direct consequence of \cref{lem:load} we obtain the following claim.
\begin{corollary}\label{corollary:frequency-of-e}
	Consider $n \cdot t$ independent random walks of length $\ell \in O(\log{n} / \eps)$ such that from each vertex there are $t$ walks originating. Then, with high probability an edge $e$ is contained in $O(t \log^2{n} / \eps)$ of those walks.
\end{corollary}

In \cref{sec:random-walks}, we describe how to update our data structures in $O(\ell \cdot \log n)$ time per an update of an $\ell$-length walk.
Since \cref{alg:PageRank-approximation} runs $t=R=O(\log n / (\epsilon \alpha^2))$ random walks per vertex, by \cref{corollary:frequency-of-e} there are $O(\log^3n/(\eps  \alpha^2))$ walks passing through each edge. Thus by the fact that walks have lengths $O(\log n/\eps)$ with high probability, the dynamic algorithm requires $O(\log^5n/(\eps^2\alpha^2))$
 time for each update, which yields \cref{thm:undirected}.

%%%%%%%%%%%%%%%%%%%%%%%%%%%%%%%%%%%%%%%%%%%%%%%%%%%%%%%%%%%%%%%%%%%%%%%%%%%%%
%                       Bibliography                                        %
%%%%%%%%%%%%%%%%%%%%%%%%%%%%%%%%%%%%%%%%%%%%%%%%%%%%%%%%%%%%%%%%%%%%%%%%%%%%%

\bibliographystyle{alpha}
\bibliography{bibliography}

\end{document}